\documentclass[runningheads]{llncs}
\usepackage{graphicx}
\usepackage{amsmath,amssymb}
\usepackage[disable]{todonotes}
\usepackage{tikz}
\usepackage{thm-restate}
\usepackage{caption}
\usepackage{subcaption}
\usetikzlibrary{decorations}
\usetikzlibrary{decorations.pathreplacing,decorations.markings}\graphicspath{{./figs/}}

\newcommand*{\calL}{\ensuremath{\mathcal{L}}}	

\begin{document}
\title{Consistent Simplification of\\ Polyline Tree Bundles}
%
%
\author{
Yannick Bosch\inst{1}\and
Peter Schäfer\inst{1}\and
Joachim Spoerhase\inst{2,3}\and
Sabine Storandt\inst{1}\and
Johannes Zink\inst{2}
}
\authorrunning{Y. Bosch et al.}
%
\institute{
University of Konstanz, Germany\and
University of Würzburg, Germany\and
University of Aalto, Finland
}
\maketitle              
\begin{abstract}
The \textsc{Polyline Bundle Simplification} (PBS) problem is a generalization of the classical polyline simplification problem. Given a set of polylines, which may share line segments and points, PBS asks for the smallest consistent simplification of these polylines with respect to a given distance threshold. Here, consistent means that each point is either kept in or discarded from all polylines containing it. In previous work, it was proven that PBS is NP-hard to approximate within a factor of $n^{\frac{1}{3}-\varepsilon}$ for any $\varepsilon > 0$ where $n$ denotes the number of points in the input. This hardness result holds even for two polylines. In this paper we first study the practically relevant setting of planar inputs. While for many combinatorial optimization problems the restriction to planar settings
makes the problem substantially easier, we show that the
inapproximability bound known for general inputs continues to hold even for planar inputs. We proceed with the interesting special case of PBS where the polylines form a rooted tree. Such tree bundles
naturally arise in the context of movement data visualization. We prove that optimal simplifications of these tree bundles can be computed in $\mathcal{O}(n^3)$ for the Fr\'echet distance and in $\mathcal{O}(n^2)$ for the Hausdorff distance (which both match the
computation time for single polylines). Furthermore, we present a greedy heuristic that allows to decompose polyline bundles into tree bundles in order to make our exact algorithm for trees useful on general inputs. The applicability of our approaches  is demonstrated in an
experimental evaluation on real-world data.

\end{abstract}
\section{Introduction}
Polyline simplification is a well-studied optimization problem~\cite{dou73,her92,ima88,cha96,vis93} with a wide field of applications, e.g., in computer graphics, map visualization, or data smoothing.
In the classical sense, polyline simplification means removing some polyline bend points
while keeping a small distance to the original polyline.
Given a distance threshold, the optimal simplification of a single polyline, i.e., the simplification keeping as few polyline points as possible, respecting that threshold can be computed in polynomial time \cite{ima88}. However, in case the input is a set of (partially) overlapping polylines, individual simplification of each polyline leads to visually unpleasing results as shared parts may be simplified in different ways. Moreover the visual complexity might even increase which opposes the simplification concept.
Aiming at more appealing and sensible results, the problem of \textsc{Polyline Bundle Simplification (PBS)} was introduced in \cite{spo20}.  It adds as an additional constraint that shared parts must be simplified consistently (i.e. each point is either kept in or discarded from all polylines containing it).

\begin{definition}[Polyline Bundle Simplification \cite{spo20}]
\label{def:pbs}
An instance of PBS 
is a triple $(P,\calL,\delta)$ where $P=\{p_1,\ldots,p_n\}$ is a set of $n$ points in the plane, $\calL = \{L_1,\ldots L_\ell\}$ is a set of $\ell$ simple polylines, each represented as a list of points from $P$ (here, simple means that each point appears at most once in $L_i$), and $\delta$ is a distance parameter.
The goal is to obtain a minimum size subset $P^* \subseteq P$ 
such that for each polyline $L \in \calL$ its induced simplification $L\cap P^*$ contains the start and end point of $L$ and has a segment-wise distance of at most $\delta$ to $L$.
\end{definition}

PBS is a generalization of the classical polyline simplification problem  but was proven to be NP-hard to approximate within a factor of $n^{\frac{1}{3}-\varepsilon}$ for any $\varepsilon > 0$ already for two polylines when using the Hausdorff or the Fr\'echet distance \cite{spo20}.
Motivated by this strong hardness result, the goal of this paper is to investigate the complexity of practically interesting special cases of PBS and to design and evaluate practical algorithms.
\bigskip\\
\textbf{Related Work.} Simultaneous simplification of multiple polylines was considered in previous work e.g. in the context of computational biology or for map generation. The so called \emph{chain pair simplification problem} asks for two polylines for their simplifications such that for given $k \in \mathbb{N}$ and $\delta >0$ each simplified chain contains at most $k$ segments, and the Fr\'echet distance between them is at most~$\delta$~\cite{ber08}. The problem arises in protein structure alignment or map matching tasks and was studied from a theoretical and practical perspective \cite{wyl13,fan15,fan16}.
While the basic idea to preserve resemblance between polylines after simplification is similar to the motivation behind PBS, chain pair simplification only ever considers two polylines  and does not put further restrictions on the simplification of shared parts.
Analyzing bundles of (potentially overlapping and intersecting) movement trajectories is an important means to study group behavior and to generate maps. For example, the RoadRunner approach \cite{he18} infers high-precision maps from GPS trajectories. In~\cite{buc18}, an approach was proposed that computes a concise graph that represents all trajectories in a given set sufficiently well. But these and similar methods do not produce valid simplifications of each input polyline, but allow to discard outliers or to let a polyline be represented by a completely disjoint polyline which is quite different from the PBS setting.

The PBS problem was introduced in \cite{spo20}.
In addition to the  above mentioned inapproximability result, there were also two algorithms for PBS discussed in the paper. For PBS with the Fr\'echet distance, a bi-criteria $(O(\log(\ell+n)),2)$-approximation algorithm was presented. This algorithm is allowed to return results within a distance threshold of $2\delta$, and based on this constraint relaxation achieves a logarithmic approximation factor (compared to the optimal solution for $\delta$) in polynomial time. Furthermore, it was shown that PBS is fixed-parameter tractable in the number $k$ of points that are shared by at least two polylines, based on a simple algorithm with a running time of $O(2^k \cdot \ell \cdot n^2 + \ell \cdot n^3)$.
\bigskip\\
\textbf{Contribution.}
We present the following new theoretical and practical results:
\begin{itemize}
\item PBS remains NP-hard to approximate to within a factor~$n^{\frac{1}{3}-\epsilon}$ for any $\epsilon>0$ on \emph{planar} inputs (Section \ref{SEC:planarhard}).
\item The special case of PBS where the polylines form a rooted tree can be solved optimally in polynomial time. Similar to the Imai-Iri algorithm for simplification of a single polyline \cite{ima88}, our algorithm precomputes the possible set of shortcuts for the given distance threshold and thereupon transforms the given geometric problem  into a graph problem.  But while in the Imai-Iri algorithm a simple search for the minimum link-path in the shortcut graph suffices, we need a more intricate dynamic programming approach (DP) to deal with the tree structure (Section \ref{SEC:treebundles}).
\item We devise a greedy heuristic that decomposes a general polyline bundle into tree bundles, which then can be simplified independently and optimally  with our DP (Section \ref{SEC:decomposition}).
\item In the experimental evaluation, we use our new approach to simplify polyline bundles that model movement data or public transit maps. We compare our approach in terms of efficiency and quality to the bi-criteria approximation algorithm proposed in \cite{spo20}  (Section \ref{SEC:experiments}).
\end{itemize}

\section{Preliminaries}
The two most commonly used distance functions $d$ to govern polyline simplification are the Fr\'echet distance $d_F$ and the Hausdorff distance $d_H$.
In the context of polyline simplification, the distance function $d$ is used to measure the distance of a line segment $(a,b)$ in the simplification to the corresponding sub-polyline of $L$, which we abbreviate by $L(a,b)$.
For any line segment $(a,b)$ in a valid simplification, we require $d((a,b), L(a,b)) \leq \delta$.
\todo{COCOON reviewer 2 suggests: In practical polyline simplification, one often draws a "tube" of width delta around the input polyline and requires the simplification to stay inside. How is this reflected in the Frechet/Hausdorff setting?}%
Given a single polyline $L$ of length $n$, a distance function $d$, and a distance threshold $\delta$, an optimal simplification of $L$ can be computed in time $\mathcal{O}(n^3)$ using the Imai-Iri algorithm \cite{ima88}. The algorithm starts by constructing a so called shortcut graph, in which there is an edge between pairs of points $a, b$ in $L$ if $d((a,b), L(a,b)) \leq \delta$. Checking this property for each of the $\Theta(n^2)$ point pairs takes $\mathcal{O}(n^3)$ time when using the Fr\'echet or the Hausdorff distance. In the created shortcut graph, the best simplification can be identified by computing the minimum-link path between the start and the end node of $L$ with a BFS run. As this only takes time linear in the graph size,  the shortcut graph computation dominates the overall running time.

An impoved method for shortcut graph construction presented by Chan and Chin \cite{cha96} can reduce the running time to $\mathcal{O}(n^2)$ for $d = d_H$. The algorithm  uses sweeps to first compute directed shortcuts from which then subsequently the valid undirected shortcuts can be deduced. More precisely, in each sweep, every point $p$ is considered as possible starting point of a directed shortcut in forward direction. To then efficiently decide whether a later point on the polyline is a valid endpoint of such a directed shortcut, a cone is maintained  in which all valid endpoints have to lie. Updating that cone and making the containment check is possible in constant time, hence all directed shortcuts starting in $p$ can be computed in $\mathcal{O}(n)$. The respective total time for considering every point as starting point in both sweeps is $\mathcal{O}(n^2)$. A valid undirected shortcut exists if and only if both of its directed versions are constructed in the sweep phase. This obviously can be checked for each potential shortcut in constant time, leading to an overall shortcut graph construction time of $\mathcal{O}(n^2)$.

\section{Hardness of Approximating Planar Polyline Bundle Simplification}\label{SEC:planarhard}
In this section, we show that we cannot approximate polyline bundle simplification on planar inputs
by the same polynomial factor that was
previously shown for general, non-planar inputs~\cite{spo20}.
Here, planar means that no two polyline segments touch or intersect each other unless they
share a common endpoint.

\begin{theorem}
	\label{CLM:approx-hardness}
	PBS with a planar polyline bundle as input is NP-hard to approximate within a factor of~$n^{\frac{1}{3} - \varepsilon}$ for any $\varepsilon > 0$,
	where $n$ is the number of points in the polyline bundle.
\end{theorem}
We build upon the hardness reduction from minimum independent dominating set (MIDS) from ~\cite{spo20} by modifying their gadgets and the arrangement of their gadgets such that the constructed polyline bundle is planar.
In the MIDS problem,
we are given a graph $G = (V, E)$ with $\hat{n}$ vertices and $c \hat{n}$ edges,
and the task is to find a set~$S$ of vertices
that is independent (no two vertices in~$S$ are adjacent) and dominating (each vertex not in $S$ has a neighbor in~$S$).
This problem has been shown to be NP-hard to approximate within a factor of~$\hat{n}^{1 - \varepsilon}$ for any $\varepsilon > 0$~\cite{hal93},
even for constant~$c$, i.e., sparse graphs.

\begin{figure}[t]
	\centering
	\begin{subfigure}[b]{0.12 \linewidth}
		\centering
		\includegraphics[page=2,scale=0.9,trim=274 138 815 0,clip]{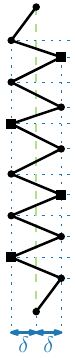}
		\caption{Vertex gadget.}
		\label{FIG:MIDS-reduction-vertex-gadget}
	\end{subfigure}
	\hfill
	\begin{minipage}[b]{.86\textwidth}
		\centering
		
		\begin{subfigure}[b]{1 \linewidth}
			\centering
			\includegraphics[page=15,scale=1.0,trim=266 232 764 33,clip]{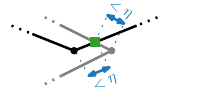}
			\caption{Crossing for planarization; a new polyline point (green square) is inserted close to polyline points of a vertex gadget (gray) and an edge or neighborhood gadget (black).}
			\label{FIG:MIDS-reduction-crossing-gadget}
		\end{subfigure}
		
		\bigskip
		
		\begin{subfigure}[b]{1 \linewidth}
			\centering
			\includegraphics[page=3,scale=0.8,trim=294 264 544 14,clip]{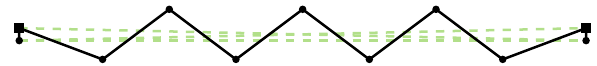}
			\caption{Edge gadget; two polyline points are shared with vertex gadgets.}
			\label{FIG:MIDS-reduction-edge-gadget}
		\end{subfigure}
		
		\bigskip
		
		\begin{subfigure}[b]{1 \linewidth}
			\centering
			\includegraphics[page=11,scale=.8,trim=332 40 432 250,clip]{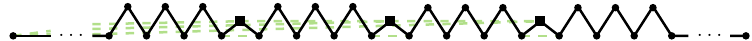}
			\caption{Neighborhood gadget; the middle points are shared with vertex gadgets.}
			\label{FIG:MIDS-reduction-neighborhood-gadget}
		\end{subfigure}
	\end{minipage}
	\caption{Schematization of the gadgets of the reduction from MIDS to PBS with planar instances. Shortcuts are indicated by dashed green line segments and points that are shared between two gadgets are drawn as squares.
	}
	\label{FIG:MIDS-reduction}
\end{figure}

The construction uses vertex gadgets allowing exactly one shortcut; see Figure~\ref{FIG:MIDS-reduction-vertex-gadget}.
Taking this shortcut represents that the corresponding vertex of the minimum independent dominating set instance is \emph{not} included in~$S$.
Moreover it uses edge gadgets connecting for each edge its two corresponding vertex gadgets.
They are comprised of long
zizag pieces that can only be skipped if the independent set property for each edge is fulfilled; see Figure~\ref{FIG:MIDS-reduction-edge-gadget}.
Similarly, we have a neighborhood gadget connecting for each vertex the vertex gadgets of its neighborhood to ensure that the domination property is satisfied; see Figure~\ref{FIG:MIDS-reduction-neighborhood-gadget}.

\begin{figure}[bt]
	\centering
	\includegraphics[page=10,width=\linewidth,trim=0 107 0 0,clip]{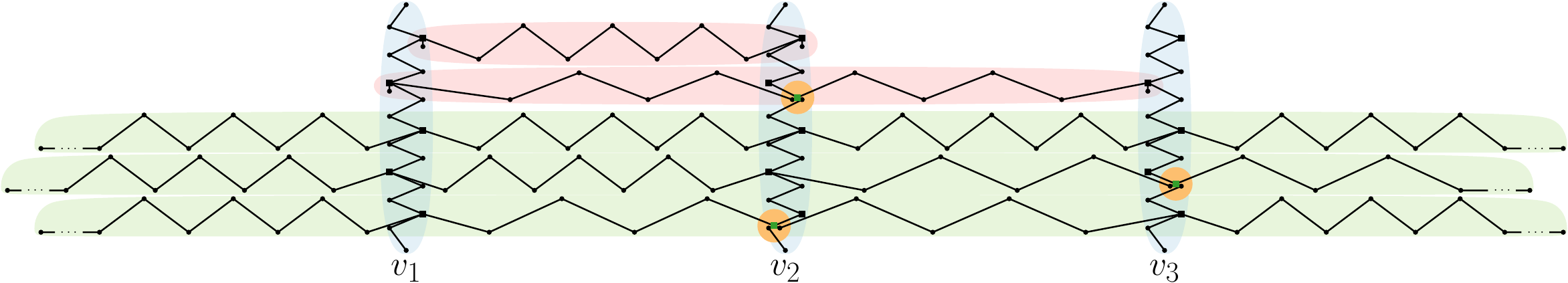}
	\caption{Combination of three vertex gadgets (for the vertices $v_1, v_2, v_3$; blue background) with two edge gadgets (for the edges $v_1 v_2$ and $v_1 v_3$; red background) and three neighborhood gadgets (for the vertices $v_1, v_2, v_3$; green background).
		We use a crossing as in~Figure~\ref{FIG:MIDS-reduction-crossing-gadget} between the vertex gadget of $v_2$ and the edge gadget of $v_1 v_3$, between the vertex gadget of $v_3$ and the neighborhood gadget of $v_2$, and between the vertex gadget of $v_2$ and the neighborhood gadget of $v_3$}
	\label{FIG:MIDS-reduction-combined-example}
\end{figure}

On a high level, vertex gadgets are vertical pieces arranged horizontally next to each other
and edge and neighborhood gadgets are horizontal pieces arranged vertically above each other and across the vertex gadgets.
This yields a grid-like structure with many crossings between
vertex gadgets and unrelated edge/neighborhood gadgets; 
see Figure~\ref{FIG:MIDS-reduction-combined-example}.
The key idea is to planarize the non-planar construction by replacing
crossings by new polyline points.  However, we have to be careful
where to insert these new points.  Just inserting points wherever a
crossing occurs would allow new shortcuts and hence destroy the
mechanics of the gadgets.  We can prevent this from happening by
reshaping the construction so that crossings occur only close to
existing polyline points.  There, we can insert new polyline points
onto the crossings sufficiently close to existing other polyline
points.  This ensures that for any shortcut starting or ending at a
crossing point, this crossing point could either be replaced by an
original polyline point or that the total saving incurred by the
new crossing points does not severely impact the gap in the objective
function between completeness and soundness in the hardness proof. We describe in the appendix in detail, how to reshape the gadgets in order to ensure correctness.
Our basic reduction uses one polyline per gadget.
But we can connect all vertex gadgets to one polyline and all edge and neighborhood gadgets to another polyline, which means our results hold true even for only two polylines (the connections have to be made carefully to not violate planarity as discussed in the appendix).
\begin{corollary}
	PBS with a planar polyline bundle as input is not fixed-parameter tractable (FPT) in the number of polylines~$\ell$.
	In particular, PBS with two polylines in a planar polyline bundle is 
	already NP-hard to approximate within a factor of~$n^{\frac{1}{3} - \varepsilon}$ for any $\varepsilon > 0$.
\end{corollary}

\section{Simplification of Polyline Tree Bundles}
\label{SEC:treebundles}
With the general PBS problem being hard to approximate better than $n^{\frac{1}{3}}$ even on planar inputs, we now consider tree bundles as another interesting special case of PBS. To form a tree bundle, the polylines have to start in a common root point and then branch out. 
\begin{definition}[Polyline Tree Bundle (PTB)]
An instance of PTB is a PBS instance $(P,\calL,\delta)$
where we additionally require that $\calL$ is a set of simple polylines such that all $L \in \calL$  start at the root point $p_r$, and for any pair of polylines $L, L' \in \calL$, the only intersection is a common prefix $L(p_r,p_i)=L'(p_r,p_i)$.
\end{definition}
We remark that this definition does not demand the tree bundle to be planar as the intersection constraint is only concerned with common points of the polylines. 
Moreover, we do not need to consider the case here where a polyline $L' \in \calL$ is a sub-polyline of $L \in \calL$. By definition, we will include the endpoints of all polylines in our simplification and hence if the endpoint of $L'$ lies on $L$, we could simply consider that point as the root of another PTB which can be simplified independently.  We will show that tree bundles can be consistently  simplified to optimality  in polynomial time. 
\smallskip\\\textbf{Problem Transformation.}
We transform the PTB simplification problem into a graph problem by constructing two directed graphs from the input data: a \emph{tree graph} and a \emph{shortcut graph}. We start by considering the polylines as embedded directed paths which start at the root point.
The tree graph  $G_\textnormal{t}=(V,E_\textnormal{t})$ is the union of these paths. 
More precisely, for each point occurring in the PTB there is a corresponding node $v \in V$ (with $v_r$ corresponding to the root point $p_r$), and there exists a directed edge $(v,w) \in E_\textnormal{t}$ if there is a polyline $L \in \calL$ which contains the  segment between the respective points (in that direction). For a given distance function $d$ and threshold $\delta > 0$,  the shortcut graph $G_\textnormal{s} = (V,E_\textnormal{s})$ is the union of all valid shortcut edges, i.e. edges $(v,w) \in \binom{V}{2}$ where for all  polylines $L \in \calL$ that contain $v$ and $w$ (in that order), we have $d((v,w),L(v,w)) \leq \delta$. Figure~\ref{fig:PTB} shows  $G_\textnormal{t}$ and $G_\textnormal{s}$ for an example PTB.  Note that no matter the distance function and the value of $\delta$, we always have  $E_\textnormal{t} \subseteq E_\textnormal{s}$, i.e. all tree graph edges are also contained in the shortcut graph.

\begin{figure}[b]
   \includegraphics[width=.45\textwidth]{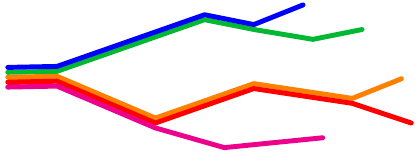}\hfill
   \includegraphics[width=.45\textwidth]{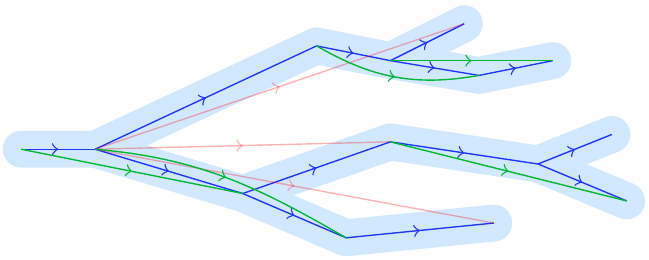}
    \caption{Left: Example of a PTB instance. Right: Blue edges represent the tree graph $G_\textnormal{t}$. The combination of the blue and green edges build the shortcut graph $G_\textnormal{s}$ for $d_H$ for the distance threshold $\delta$ (indicated via the light blue tubes). Examples of invalid shortcuts  are drawn in red.}
    \label{fig:PTB}
\end{figure}

\begin{restatable}[]{lemma}{lemTreeGraph}\label{LEM:treegraph}
The tree graph $G_\textnormal{t} = (V,E_\textnormal{t})$ has size $\mathcal{O}(n)$ and can be constructed in time~$\mathcal{O}(n^2)$.
\end{restatable}
\begin{restatable}[]{theorem}{thmHausdorff}\label{THM:hausdorff}
The shortcut graph $G_\textnormal{s} = (V,E_\textnormal{s})$ has size $\mathcal{O}(n^2)$ and can be constructed for the Fr\'echet distance in time $\mathcal{O}(n^3)$ and for the Hausdorff distance in time $\mathcal{O}(n^2)$.
\end{restatable}
The respective proofs are provided in the appendix. Based on the notion of the tree graph and the shortcut graph,
we are now ready to restate the PTB simplification problem (PTBS) as a graph problem.
\begin{definition}[Polyline Tree Bundle Simplification (PTBS)]
Given a tree graph $G_\textnormal{t} = (V,E_\textnormal{t})$ and a shortcut graph $G_\textnormal{s} = (V,E_\textnormal{s})$, the goal is to find a smallest node subset $S \subseteq V$ such that:
\begin{itemize}
\item The root node and all leaf nodes of the tree graph are contained in $S$.
\item The induced subgraph $G_\textnormal{s}[S]$ is connected.
\end{itemize}
\end{definition}
\vspace*{0.25cm}
\textbf{Exact Polytime Algorithm.}
Next, we describe a dynamic programming (DP) approach that only operates on $G_\textnormal{t}$ and  $G_\textnormal{s}$, and returns an optimal PTBS solution in time $\mathcal{O}(n^2)$. 
Let $Sub(v) \subseteq G_\textnormal{t}$ be the sub-tree rooted at node $v$ in the tree graph. Our main observation is that we can break down an optimal solution recursively. If a node $v$ is part of the solution, it's easy to see that there can't be shortcuts bypassing $v$. Thus, the solution $S$ can be split into two parts: an optimal solution for $Sub(v)$ and an optimal solution for $G_\textnormal{t} \setminus Sub(v)$.
We denote the size of an optimal solution for $Sub(v)$ by $s(v)$. As we don't know a priori which nodes will end up in the solution, we strive for computing  $s(v)$ for each node $v \in V$ in an efficient manner. For leaf nodes $v$, we obviously get $s(v)=1$. To compute $s(v)$ for an inner node $v$, we assume that $s(w)$ is already known for all nodes $w \in Sub(v)\setminus\{v\}$. Each path from a leaf $u$ to $v$ in $Sub(v)$ needs to contain a \emph{cover node} $w$  such that $(v,w) \in E_\textnormal{s}$  (that means there is a valid shortcut from $v$ to $w$). To identify the best selection of such  cover nodes, we compute a helping function $h : V \rightarrow \mathbb{N}$ for each node $w \in Sub(v)$ as follows:  Initially, $h(w)=s(w)$ if $(v,w) \in E_\textnormal{s}$, and $h(w) = \infty$ otherwise. Then, in a post-order traversal of $Sub(v)$, for each non-leaf node $w$ we set $h(w) = \min \{h(w), \sum_{u \in N(w)} h(u)\}$  where $N(w)$ denotes the set of children (out-neighbors) of $w$ in $G_\textnormal{t}$. In that way, $h(w)$ encodes the smallest number of nodes that have to be kept in $Sub(w)$ if for all paths from $v$ to leaf nodes in $Sub(w)$ the respective cover node is contained in $Sub(w)$. The optimal solution size $s(v)$ for $Sub(v)$ is then $h(v)+1$ (as we have to additionally include $v$ itself). Note that $s(v)$ is always well-defined (i.e., finite) as the tree edges are all valid shortcuts in $G_\textnormal{s}$. To make sure that at the time we compute $s(v)$ all values $s(w)$ for $w \in Sub(v) \setminus \{v\}$ are known, we also globally traverse the nodes in the tree graph in post-order.   Figure \ref{FIG:tree_DP} illustrates the computation of $s(v)$. The optimal set of simplification nodes can then be determined by backtracking. 
\begin{figure}[t]
\centering
\includegraphics[width=0.8\textwidth]{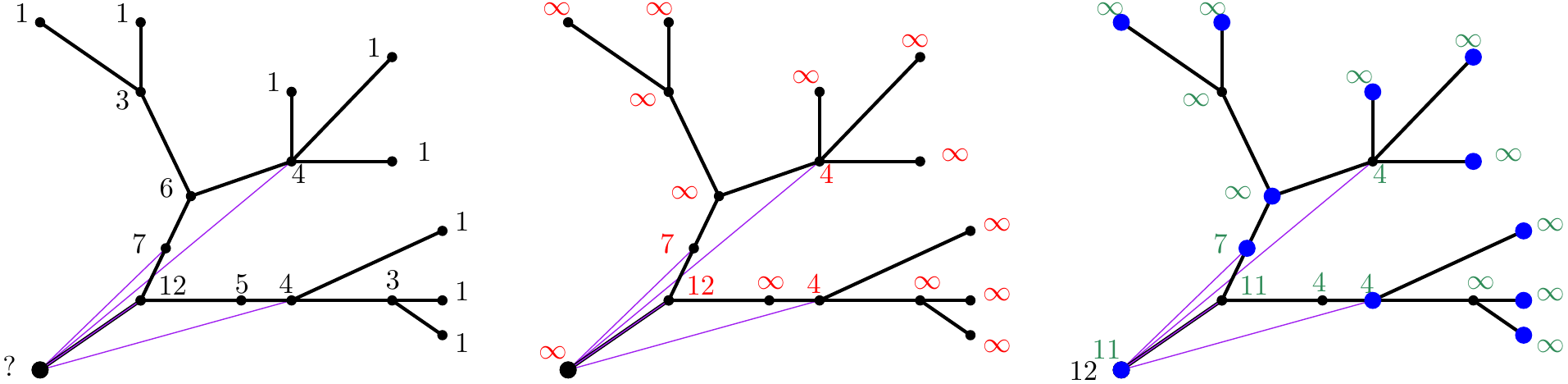}
\caption{The left image shows an example  tree graph with optimum sub-tree simplification sizes (black) known for all nodes except the root node. The purple line segments indicate valid shortcuts from the root node. The middle image depicts the same tree after initial assignment of the helping values (red). Here, only end nodes of valid shortcuts have finite values assigned to them. The right image shows the final helping values (green) after propagation as well as the respective optimum simplification size of the tree assigned to the root node (black). The blue marked nodes are the ones  that are contained in the optimal simplification.  }\label{FIG:tree_DP}
\end{figure}

For a faster running time of the DP in practice (used in our experiments), we only compute $h$-values for nodes in $Sub(v)$ which are on a path from $v$ to some node $w$ with $(v,w) \in E_\textnormal{s}$. These nodes can  easily be identified by computing the reverse path from each such node $w$ to $v$ and marking all nodes along the way (stopping as soon as a marked node is encountered to avoid redundancy). For marked nodes $w$ with an unmarked neighbor, we just set $\sum_{u \in N(w)} h(u)$ to $\infty$ to maintain correctness. Especially for small distance thresholds $\delta$ and large sub-trees $Sub(v)$, this modification accelerates the computation of $s(v)$ significantly.

\begin{restatable}[]{theorem}{thmCorrect}\label{THM:correct}
PTBS can be solved optimally in time $\mathcal{O}(n^2)$.
\end{restatable}
The proof is given in the appendix.
%
Combining the time for problem transformation with  the time of the DP, we get an overall running time of $\mathcal{O}(n^3)$ for PTBS when using the Fr\'echet distance and a running time of $\mathcal{O}(n^2)$ when using the Hausdorff distance. Hence -- although having to use more complicated machinery -- we end up with running times for tree bundle simplification that match the best known running times for simplification of a single polyline.

\section{Tree Bundle Decompositions}\label{SEC:decomposition}
To leverage our algorithm for  optimal tree bundle simplification for general bundles, we next consider the problem of decomposing a general bundle into  (a small set of) tree bundles. To formalize the \textsc{Tree Bundle Decomposition} (TBD) problem we first introduce the notion of a $D$-decomposition of a polyline.
\begin{definition}[$D$-Decomposition]
Let $L=(s, \dots, t)$ be a simple polyline (represented as a list of points) and let $D$ be a point set. Further let $d_1, d_2, \dots, d_k$ be the points in $L  \cap D$ in the order in which they appear in $L$. The $D$-decomposition of $L$ denoted by $L(D)$ is the set of subpolylines $L(d_i,d_{i+1})$ for $i=1,\dots, k-1$.
\end{definition}
We strive to find a sensible  set $D$ that partitions a given bundle into tree bundles. \vspace{-8pt}%
\begin{definition}[Tree Bundle Decomposition (TBD)]
\label{def:tbd}
Given a PBS instance $(P,\calL,\delta)$,  we seek to find a point subset $D \subseteq P$ (the decomposition points) with the following requirements:
\begin{itemize}
\item Each polyline $L \in \calL$ starts and ends in a point in $D$. 
\item Let $G_I$ be the intersection graph in which we have a node for each subpolyline in $\bigcup_{L \in \calL} L(D)$ and an edge between two nodes if the subpolylines $\tilde{L}, \bar{L}$ share a point that is not in $D$, i.e. $(\tilde{L} \cap \bar{L}) \setminus D \neq \emptyset$. Then the subpolylines within a connected component in $G_I$ form a PTB.
\end{itemize}
\end{definition}
Based on a TBD, we can simplify the given  bundle by  simplifying each of the  tree bundles induced by $D$ independently. The union of all tree simplifications then yields $S$. The goal, of course, is still to end up with a small set $S$. To achieve that, we aim at TBDs which induce few but large tree bundles with a small decomposition set $D$. In the following we assume that all polyline endpoints are already included in $D$ as they have to be part of $S$ by definition.
\smallskip\\\textbf{A Simple Greedy Heuristic.}
Nodes in the set $D$ might end up being the root node of a tree bundle or a leaf node (or both). One way to construct $D$ is hence to greedily  select root nodes and grow trees from those (adding the respective leaf nodes to $D$ as well). 

For root selection, we use the line degree of the point, that is, the number of polylines in $\mathcal{L}$ that contain the point. As only polylines that contain the root  can be part of the respective PTB, we always choose  the node with the highest line degree that is not already part of a tree bundle next.
To compute the largest prossible tree bundle for a selected root node $r$, we first construct the union graph $G_U(V,E)$ of the polyline bundle.  Here,  each point in the bundle is represented by a node in $V$ and an edge exists between two nodes if there is a polyline segment  between the respective points. Additionally, we assign to each edge in $G_U$ the set of polylines that traverse it. A tree bundle can then be computed in a BFS-like fashion in $G_U$ starting from $r$, always pushing edges instead of nodes in the queue. The edges incident to $r$ are always included in the PTB and are hence used for initialization of the queue (artificially directed away from $r$). In any later step, if an edge $(u,v)$ is extracted from the queue, we first check whether all other edges incident to $v$ are unvisited. If that is the case, we need to make sure that the polyline set assigned to each incident edge is a (not necessarily proper) subset of the polylines assigned to $(u,v)$. If and only if those conditions are met for all incident edges, these edges are included in the subtree, marked as visited, and inserted in the queue. Otherwise $v$ is added to $D$. The process takes $\mathcal{O}(\ell\cdot n)$ time.

\section{Experimental Evaluation}\label{SEC:experiments}
We implemented the dynamic programming  approach (DP) for exact tree bundle simplification as well as  the greedy tree bundle decomposition algorithm (TBD)  in C++. Furthermore, we also provide the first implementation of the bi-criteria approximation algorithm (BCA) from \cite{spo20}. BCA demands to first compute a small \emph{star cover} of the polyline bundle where a star is a point $p$ in the bundle together with selected shortcuts that end in $p$. A feasible star cover has to ensure that for each polyline $L \in \mathcal{L}$ and for each segment in $L$, there is a  star in the  cover with a shortcut that bridges said segment. The set of points of all stars in the star cover induces a simplified polyline bundle  for a distance threshold of $2\delta$ (that is, twice the actual threshold).
The number of retained points is at most a factor of $\mathcal{O}(\log(\ell+n))$ larger than the optimal solution for threshold $\delta$.
The running time of BCA is in $\mathcal{O}(\ell \cdot n^3)$.
As this result only holds when using the Fr\'echet distance, we will focus in the experiments on  $d_F$. All experiments were run on a single core of an Intel Core i9 processor at 2.4 GHz.
\smallskip\\\textbf{Benchmark Data.}
We used two types of polyline bundle data to evaluate the algorithms: \textit{(i) Path bundles from embedded road networks} (extracted from OpenStreetMap \cite{osm}). Such bundles are a good model for movement data. Bundles were constructed by first extracting a connected subgraph with a given number of nodes from the network. To obtain a tree bundle,  we then  performed a  BFS run from a randomly selected root node in the subgraph and backtracked all paths from the leaves to the root of the BFS-tree.  For general bundles, we select not one but several root nodes in the subgraph, construct a tree bundle for each and then combine those into a single bundle.
\todo{COCOON reviewer 2 says: ``The data extracted for part (i) general bundles is still quite biased to be nicely decomposable into trees. This suits your algorithm, but I don't see an argument why this would be representative of real-world instances? And if so, certainly for very specific trajectories. It is no surprise that TBD+DP performs very well on this type of data.''}
\textit{(ii) Public transit networks} (GTFS data provided by OpenMobilityData \cite{omd}). We used the data from Stuttgart, Freiburg, Manhattan and Chicago. Here each bus or train line forms a polyline in our bundle.
\smallskip\\
\textbf{Tree Bundle Simplification Results.}
We compared the performance of DP and BCA on tree bundles of different sizes extracted from road networks. While it might seem to be an apples-to-oranges comparison, when we have an exact algorithm on the one side and a bicriteria approximation on the other, it is not a priori clear which algorithm would produce the smaller simplification when tested with the same $\delta$ (as BCA is allowed to exceed it by a factor of 2). We observe, however, that on all tested instances, the exact DP algorithm produces better simplification results than BCA, even though BCA is allowed to use a distance threshold of $2\delta$. If we call BCA with $\delta/2$ to then end up with a solution that obeys the $\delta$-constraint, the quality deteriorates significantly (with up to 50\% larger outputs).  Table \ref{TAB:treebundle-results} provides some selected results which reflect the general behavior. It is interesting  that the BCA algorithm indeed produces solutions where the $\delta$ threshold is violated by a factor of 2, proving the theoretical analysis to be tight in this respect. We also observe that the DP approach scales much better, with running times up to a factor of 50 faster than BCA on our largest test instance. 
\hyphenation{re-sul-ting}
\begin{table}[t]
\begin{minipage}{0.58\textwidth}
	\begin{tabular}{|l|r|r|r|r|r|r|}
	\hline
	 & $\delta \cdot 10^4$ & $\delta_F\cdot 10^4$ & $\delta_F/\delta$ & $n$ & $|S|$ & time \\
	\hline\hline
	DP  &  5.00  & 4.99 & 0.99 &  500 & \textbf{204} & 3\\
	\hline
	BCA & 5.00  & 9.81 & 1.96  &  500 &  {216} & 3 \\
	\hline
	BCA &   2.50 & 3.17 & 1.27  & 500 & {251} & 6\\
	\hline\hline	
	DP  & 5.00 & 5.00  & 1.00 &  8,000    &  \textbf{4009}& 21\\
	\hline
	BCA  & 5.00 & 8.77& 1.75 & 8,000 &  {4029}& 407\\
	\hline
	BCA  & 2.50& 4.04& 1.61 & 8,000&  {5276}& 350\\
	\hline\hline	
	DP  & 5.00 & 5.00  & 1.00  & 50,000   & \textbf{24,076} & 248 \\ 
	\hline
	BCA & 5.00 & 9.68& 1.94 & 50,000 & {24,195} & 14,800 \\
	\hline 
	BCA & 2.50 & 5.00& 2.00   & 50,000   & {32,457} & 13,500 \\
	\hline
	\end{tabular}
	\end{minipage}\hfill
	\begin{minipage}{0.4\textwidth}
		\caption{Comparison of DP and BCA  on tree bundles (note that BCA is tested for  the original $\delta$ and $\delta/2$). $\delta_F$ denotes the resulting Fréchet distance and $\delta_F/\delta$ the   distance relative to the threshold. $n$ is the input size, and $|S|$  the number of points in the computed solution  for threshold $\delta$ (in geo coordinates).  Timings are in milliseconds.}\label{TAB:treebundle-results}
	\end{minipage}
	\vspace*{-0.75cm}
\end{table}
\smallskip\\
\textbf{Results on General Bundles.} We used path bundles from road networks as well as public transit networks to evaluate the performance of TBD+DP and BCA.
Again, BCA results are allowed to exceed the distance threshold $\delta$ by a factor of 2. This slack is indeed strongly exploited also on public transit networks as confirmed by our detailed BCA experiments reported in the appendix. We now focus on a comparative evaluation. We observe that our heuristic approach of first computing a tree decomposition and then simplifying the resulting trees individually is always faster than BCA, computing results within a second even for roadnetwork bundles with around 10,000 nodes while BCA takes 30 times longer. In terms of quality, TBD+DP produce comparable or even better results than BCA on the Stuttgart and Freiburg network, and clearly superior results on road network bundles. Detailed results and illustrations  are provided in the appendix. The  instances on which TBD+DP was outperformed by BCA in terms of simplification size are bundles with large grid-like structures as the Chicago and the Manhattan public transit network. Here, our tree decomposition results in a huge set of trees of which we need to keep all root and leaf nodes in the simplification. A post-processing step in which for each point in $S$, we test whether it could be removed without constraint violation could help to close that gap. 

But especially for large instances, the simplicity and the fast computation  time of TBD+DP is a great advantage over BCA; in particular as the TBD is independent of $\delta$ and individual tree simplification can be easily parallelized for further improvement. 

\section{Future Work}
Based on our finding that the bi-criteria approximation algorithm  indeed exceeds the distance threshold bound by a factor of 2 on practical instances but produces high-quality solutions, future work could investigate whether improved bi-criteria approximation factors can be proven. Furthermore, it might be interesting to investigate the existence of FPT algorithms for PBS for suitable parameters. Our results that tree bundles can be processed in polynomial time might hint at parameterizability by e.g. the treewidth of the union graph of the polylines. On the practical side, further development of  heuristics for PBS or the consideration of non-simple polylines could be sensible avenues for future work.

\bibliographystyle{splncs04}
\bibliography{references}

\newpage
\section{Appendix}
For clarity, we may refer to the points of $P$ also as \emph{polyline points}.
If a line segment between two polyline points or a simplified polyline has a segment-wise distance of at most $\delta$ to its original original counterpart,
we also call it a \emph{valid} shortcut or simplification, respectively.

\subsection{Approximation Hardness of Planar Polyline Bundles}
We now provide missing details for the reduction from MIDS to planar PBS from Section \ref{SEC:planarhard}.

We describe below in more detial how to re-shape the original gadgets from \cite{spo20} but provide a high-level description first:  We modify the vertex gadgets
such that the $y$-coordinate of the points before and after a shared
point in a vertex gadget is the same as for the upper and lower row of
points of the zigzag pieces in the edge and neighborhood gadgets.  We increase the
distance between each two vertex gadgets to stretch the the zigzag
pieces of the edge/neighborhood gadgets so that all crossings occur
close to polyline points in all affected polylines.  It is a crucial
property of the original construction~\cite{spo20} that stretching
edge and neighborhood gadgets horizontally does not change the
behavior in terms of possible shortcuts.  Onto our carefully arranged
crossing, we insert polyline points to ``planarize'' the construction;
see Figures~\ref{FIG:MIDS-reduction-crossing-gadget}.
We now analyze how close to existing polyline points these crossing points need to be placed.
We require them to be strictly inside a disk of radius $\eta$ to prevent the
emergence of new shortcuts.
Intuitively, $\eta$ is chosen sufficiently small to ensure that,
given any non-shortcut $(p, q)$,
moving $p$ or $q$ within a disk of radius $\eta$ does not bring $(p, q)$ into
the $\delta$-neighborhood disk of some third polyline point $o$. More formally, we
let
\begin{equation*}
\eta = \left( \min_{\{p,o,q\} \subseteq L, L \in \calL} \{d((p, q), o) \mid d((p, q), o) > \delta\} \right) - \delta \, .
\end{equation*}
Clearly, we can determine~$\eta$ in polynomial time.
By Lemma~\ref{CLM:no-new-shortcuts-by-crossing-points},
we show that the new crossing points do not allow new shortcuts
and, hence, the functionality of the gadgets is not affected regardless
of whether we keep the crossing point and skip the neighboring points,
which we call its \emph{skip points},
or the other way around.
Observe that skip points are non-shared points.

\begin{lemma}
	\label{CLM:no-new-shortcuts-by-crossing-points}
	The set of endpoints of all shortcuts starting at a crossing point~$p$ is
	contained in the set of endpoints of all shortcuts starting
	at one of the two skip points of~$p$.
\end{lemma}

\begin{proof}	
	We prove this by contradiction.
	Let $s_1$ and $s_2$ be the two skip points of $p$.
	Suppose there is a point $q$ such that the line segment $(p, q)$ is a shortcut,
	whereas	$(s_1, q)$ and $(s_2, q)$ are no shortcuts.
	W.l.o.g., let $p, q, s_1$ be points of a polyline $L$.
	Since the Hausdorff distance is a lower bound for the Fr\'echet distance,
	we know that also $d_H((p, q), L(p, q)) \leq \delta$.
	
	For all of the gadgets, it has been shown that, wherever there is no shortcut
	between two points $p$ and $q$,
	this is because some point~$o$ between~$p$ and~$q$ has Euclidean distance
	greater than~$\delta$ to~$(p, q)$~\cite{spo20}.
	Hence in our case and by the choice of~$\eta$,
	$d((s_1, q), o) \ge \delta + \eta$.
	
	By the choice of~$p$, we know that $d_H((p, q), (s_1, q)) < \eta$.
	For any point $o$ on $L(p, q)$,
	this implies, by using the triangle inequality,
	$d((s_1, q), o) < \delta + \eta$. A contradiction.
\end{proof}

A difference remaining is that we may save a single polyline point if
we take the crossing point instead of its two skip points.  This
however, does not affect the inapproximability bound as we argue next.
Let~$\mathsf{OPT^p}$ denote the size of an optimal solution of the
constructed planar PBS instance. Moreover, let~$\mathsf{OPT^n}$ denote
the size of an optimal \emph{well-formed} solution where we do
\emph{not} allow polyline points to be placed onto the crossings. (Think
of this, as the optimum solution in the instance prior to the
planarization step where we introduce the crossing points.) The
optimum solution for the planarized instance can be turned into a
well-formed solution by replacing any crossing point with the two
nearby skip points giving $\mathsf{OPT^n} \leq 2\cdot \mathsf{OPT^p}$.
Since any well-formed solution is in particular a feasible solution to
the planarized instance, we clearly have
$\mathsf{OPT^p} \leq \mathsf{OPT^n}$.  As in the original analysis in the
unplanarized setting we can infer that there is a threshold~$T$ such that if the
input graph $G$ of  independent dominating set is a yes-instance then $\mathsf{OPT^n}\leq T $. On the other hand, if the graph is a
no-instance then $\mathsf{OPT^n} > n^{\frac{1}{3} - \varepsilon} \cdot T$.  For
the planarized case it follows that if $G$ is a yes-instance,
$\mathsf{OPT^p}\leq \mathsf{OPT^n}\leq T$, and if $G$ is a no-instance then
$\mathsf{OPT^p}\geq \mathsf{OPT^n}/2 > n^{\frac{1}{3} - \varepsilon}/2 \cdot T$.  For
increasing $n$, this gives the same inapproximability gap as for the
non-planar case.

Moreover, in the whole polyline bundle, we have slightly more points in the vertex gadgets and in the neighborhood gadgets and for the crossings than in the non-planar construction.
In the old non-planar reduction, the number of points was upper bounded by $n \leq 10 c \hat{n}^3$.
Being a bit more generous, we can upper bound the number of polyline points in the planar construction by $n \leq 30 c \hat{n}^3$; see Lemma~\ref{CLM:mids-reduction-size} in the appendix.
In any case, this constant 10 or 30 is dominated by the $\varepsilon$,
which keeps the rest of the analysis~\cite{spo20} sound
and, hence, we conclude the correctness of Theorem~\ref{CLM:approx-hardness}.

\subsubsection{Modifications in the Gadgets of the Hardness Reduction.\smallskip\\}
\noindent\hspace{-6pt}\textbf{Vertex Gadget;} see Figure~\ref{FIG:MIDS-reduction-vertex-gadget-more-details}.
Compared to the vertex gadget in the non-planar case~\cite{spo20},
our vertex gadget has more points in the zigzag piece.
More precisely, only each third point is a shared point (instead of each second)
and there are potentially more shared points.
This is because we won't have any two edge gadgets or neighborhood gadgets on the same height.
Namely, there are $\hat{n} (c + 1)$ shared points instead of $\hat{n}$ shared points,
where $c = |E| / \hat{n}$ is some constant.
We will show in Lemma~\ref{CLM:mids-reduction-size}
that the total number of points is still sufficiently small.

Moreover, the vertical distance between each two points in a vertex gadget is reduced such that these neighboring points are on the same height as the upper/lower row of points in the zigzags of the edge and neighborhood gadgets, respectively.

Clearly, the line segment~$m$ from the first to the last point has still Fr\'echet distance at most $\delta$ to the whole vertex gadget and is hence a shortcut.
There is still no other shortcut.
Any other potential shortcut segment would cross~$m$ at most once.
Let this crossing be~$r$.
\begin{figure}
	\centering
	\begin{subfigure}[b]{0.2 \linewidth}
		\centering
		\includegraphics[page=2,trim=274 140 770 0,clip]{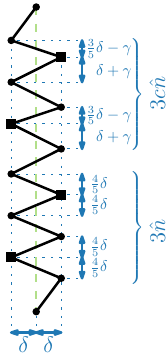}
		\caption{Vertex gadget; points that may be shared with edge or neigh\-borhood gad\-gets are drawn as squares}
		\label{FIG:MIDS-reduction-vertex-gadget-more-details}
	\end{subfigure}
	\hfill
	\begin{subfigure}[b]{.76 \linewidth}
		\centering
		\includegraphics[page=4,width=\linewidth,trim=270 222 544 0,clip]{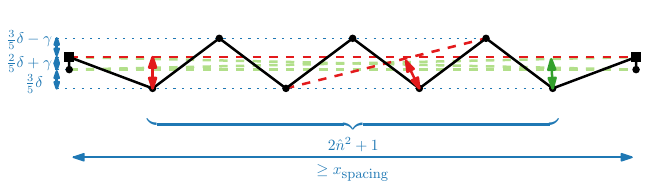}
		\caption{Edge gadget for an edge $u v$; the second and second last point (drawn as squares) are shared with the vertex gadgets of $u$ and $v$, respectively. If and only if at least one of the two shared points is skipped, we can skip all $2\hat{n}^2 + 1$ inner points.}
		\label{FIG:MIDS-reduction-edge-gadget-more-details}
	\end{subfigure}
	
	\bigskip
	
	\begin{subfigure}[b]{1 \linewidth}
		\centering
		\includegraphics[page=12,width=\linewidth,trim=322 0 430 245,clip]{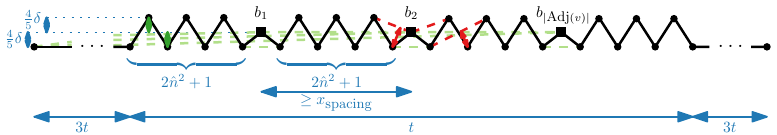}
		\caption{Neighborhood gadget for a vertex $v$; the points drawn as squares are shared with the vertex gadgets of~$v$ and $v$'s neighbors in the graph. Only if we keep at least one of the shared points, we can skip almost all points of the gadget.}
		\label{FIG:MIDS-reduction-neighborhood-gadget-more-details}
	\end{subfigure}

	\caption{More details on the gadgets in the planar hardness reduction. A green dashed segments indicates that there is a shortcut between its endpoints, whereas a red dashed segment indicates that there is no shortcut.}
	\label{FIG:MIDS-reduction-more details}
\end{figure}

For the Fr\'echet distance, we can assign points to the left of~$s$ only to the part of the shortcut segment above~$r$ and points to the right side of~$s$ only to the part of the shortcut segment below~$r$~-- or both the other way round.
In any case, since the zigzag piece jumps between the left and right side of~$m$ and no two points of a vertex gadget have the same $y$-coordinate,
there cannot be another shortcut.
For the Hausdorff distance, we can choose each third $y$-distance between two consecutive points in the zigzag sufficiently large as we do not need to obey the vertical gaps within the edge and neighborhood gadgets
(i.e. $(3/5 \delta - \gamma, \delta + \gamma)$ and $(4/5 \delta, 4/5 \delta)$, respectively) there.

The smaller $y$-distances in the zigzags also cause flatter pockets at the vertex gadgets.
However, for the minimum horizontal distance between each two vertex gadgets,
we use at least the same value~$x_\textrm{spacing}$ as before~\cite{spo20},
which is so large that the distance between each two points of a zigzag piece in
an edge or neighborhood gadget is at least~$3 \delta$.
As the width of the vertex gadget is just $2 \delta$, the zigzag pieces of the edge/neighborhood gadgets do not cross
other parts of the zigzag of a vertex gadget when leaving these pockets.

\bigskip

\noindent \textbf{Edge Gadget;} see Figure~\ref{FIG:MIDS-reduction-edge-gadget-more-details}.
The edge gadget is precisely the same as in the non-planar case, including the parameter~$\gamma$.
So we can assume that it still works as desired,
i.e. we can take a long shortcut if and only if we keep at most one of the two shared points.
This enforces the independent set property.

\bigskip

\noindent \textbf{Neighborhood Gadget;} see Figure~\ref{FIG:MIDS-reduction-neighborhood-gadget-more-details}.
The neighborhood gadget is essentially the same as in the non-planar version.
Let the corresponding vertices in horizontal order of the vertex gadgets be $u_1, \dots, u_{|\textrm{Adj}(v)|}$
and the shared points be $b_1$ and $b_{|\textrm{Adj}(v)|}$, respectively.
Between each two $b_i$ and $b_{i + 1}$ (i $\in \{1, \dots, |\textrm{Adj}(v)|-1\}$), we still add a zigzag with
$2 \hat{n}^2 + 1$ points.
Different from the non-planar construction,
we have another zigzag piece with $2 \hat{n}^2 + 1$ points before~$b_1$
and after~$b_{|\textrm{Adj}(v)|}$.
We use these additional two zigzags to cross over the vertex gadgets on the left of~$b_1$ and on the right of~$b_{|\textrm{Adj}(v)|}$.
As in the non-planar case, we define $t$ as the distance between the first point of the first zigzag piece and the last point of the last zigzag piece.
We also add a first point $3 t$ to the left of the first zigzag piece.
Symmetrically, we add a last point $3 t$ to the right of the last zigzag piece.

The claim that the only shortcuts are (i) the shortcuts skipping only $b_i$ for $i \in \{1, \dots,|\textrm{Adj}(v)|\}$ and (ii) the shortcuts starting at the first point or $b_i$ with $i \in \{1, \dots, |\textrm{Adj}(v)|\}$ and ending at the last point or $b_j$ with $i < j \in \{1, \dots, |\textrm{Adj}(v)|\}$~-- except for the shortcut starting at the first and ending at the last point~--
still holds true even though we have two additional zigzag-pieces before and after our shared points.
The critical part of the proof of this claim is that there is a shortcut from the first or last point to some $b_i$ for $i \in \{1, \dots, |\textrm{Adj}(v)|\}$.
Observe that the key argument remains true, i.e.,
the potential shortcut segment from, say, $b_1$ to the last point has a $y$-distance of
at most~$\delta$ to the upper row of points since we have chosen the $x$-distance of the last and the second last point to be $3 t$, which is sufficiently large.

Consequently, we can skip almost all points in a neighborhood gadget if we keep at least one point of $b_1, \dots, b_{|\textrm{Adj}(v)|}$.
If we skip all of them, we can skip no other point. 
So, to avoid high costs, we must not take the shortcut of the vertex gadget of at least one vertex of~$\textrm{Adj}(v)$.
This enforces the dominating set property.

By the following lemma, we show that the number of points in the new planar instance is bounded by $n \leq 30 c \hat{n}^3$.

\begin{lemma}
	\label{CLM:mids-reduction-size}
	By our reduction, we obtain from an instance $G = (V, E)$ of minimum independent dominating set an~instance of PBS with a planar polyline bundle as input that has $n \leq 30 c \hat{n}^3$ points, where $\hat{n} = |V| \geq 2$, $|E| = c \hat{n}$ ($c \geq 1$ is~constant).
\end{lemma}

\begin{proof}
	We assume that we have a shared point for each pair of vertex gadget and edge gadget, and for each pair of vertex gadget and neighborhood gadget.
	This is either a shared point because the corresponding vertex is incident to the corresponding edge/neighborhood or it is a point inserted to overcome a crossing.
	To count the points of the vertex, edge, and neighborhood gadgets, we charge the shared points to the vertex gadgets.
	So in total we have
	$\hat{n} \cdot (\hat{m} + \hat{n})$
	shared points.
	All vertex gadgets together have
	$\hat{n} (3 c \hat{n} + 3 \hat{n} + 2)$
	points, all edge gadgets have
	$\hat{m} (2\hat{n}^2 + 3)$
	unshared points, and all neighborhood gadgets have
	$(2 \hat{m} + 2) \cdot  (2\hat{n}^2 + 1) + 2 \hat{n}$
	unshared points.
	Summing these values up and using $\hat{m} = |E| = c \hat{n}$ yields (for $\hat{n} \ge 2$)
	\begin{align}
	n &\le \hat{n} \cdot (\hat{m} + \hat{n}) + \hat{n} (3 c \hat{n} + 3 \hat{n} + 2) + \hat{m} (2\hat{n}^2 + 3) + (2 \hat{m} + 2) \cdot  (2\hat{n}^2 + 1) + 2 \hat{n} \nonumber\\
	&= 6c \hat{n}^3 + (4c + 8)\hat{n}^2 + (5c + 4) \hat{n} + 2
	\leq 30 c \hat{n}^3 \, .
	\end{align}
\end{proof}

\subsubsection{Connecting the Gadgets to Two Polylines.}
At this point, we use one polyline per gadget.
So, our reduction uses $(2 + c) \hat{n}$ polylines.
We can reduce the number of polylines to two
by connecting all vertex gadgets from left to right in a row
(alternating the connection pieces between bottom and top side),
which gives us the first polyline, and by connecting all edge and neighborhood gadgets similarly,
which gives us the second polyline.
For the latter, it is a bit more subtle do this without creating new crossings or shortcuts.
The neighborhood gadgets are already relatively long and we can simply connect their endpoints, which are far away from the vertex gadgets,
without creating new shortcuts.
The edge gadgets, however, have their start and end points in between vertex gadgets.
The solution is to extend them to reach to the left and the right of all vertex gadgets similar to the neighborhood gadgets.
There, we can connect them without creating new crossings or shortcuts.
As for the neighborhood gadgets, we can do this by adding two additional zigzag pieces~--
one before the first and one after the last point of the edge gadget,
which cross all vertex gadgets to the left and the right in the way we describe above.
Observe that this also does not affect the approximation ratio asymptotically.

\newpage
\subsection{Exact Algorithm for Polyline Tree Bundle Simplification}
Next, we provide the proofs for the correctness and the running time of he dynamic programming approach for optimal tree bundle simplification. The respective lemmata and theorems are restated here for easier comprehension.

\begingroup
\def\thelemma{\ref{LEM:treegraph}}
\lemTreeGraph*
\endgroup
\begin{proof}
In a PTB on $n$ points, there can be at most $n$ polylines (based on the observation that no polyline fully contains another) with a total of $\mathcal{O}(n^2)$ segments (based on all polylines being simple). The union of the directed paths induced by the polylines can then be constructed by considering the polylines one-by-one. For a polyline, we start at its end point and then add the edges that represent the respective polyline segments backwards one after the other until we either reach the root node or a node which already was considered as a segment endpoint in another polyline (then clearly the remaining path is already part of $G_\textnormal{t}$). This process takes  $\mathcal{O}(n^2)$ time and, given the tree structure, leads to an edge set $E_\textnormal{t}$ with  $|E_\textnormal{t}| \in \mathcal{O}(n)$. 
\end{proof}

\begingroup
\def\thetheorem{\ref{THM:hausdorff}}
\thmHausdorff*
\endgroup
\begin{proof}
We compute the set of valid shortcuts by  considering the polylines one-by-one in an arbitrary order. To avoid redundant computations along shared parts, we store the result for already considered node pairs. Accordingly, the total number of potential shortcuts that need to be checked is in $\mathcal{O}(n^2)$. The time $T_d$ to check the validity of a shortcut is in $\mathcal{O}(n)$ for both $d=d_F$ and $d=d_H$. Therefore, the total construction time is in $\mathcal{O}(n^3)$.

 However, for $d_H$ we can do better by leveraging the sweep method from Chan and Chin for shortcut computation for single polylines in time $\mathcal{O}(n^2)$ \cite{cha96}. But we cannot apply that method to each polyline individually, though,  as then the running time would be again in $\mathcal{O}(n^3)$. Hence to get an improvement, we have to avoid redundant computations along shared parts.

 We first consider the shortcuts directed towards the root point of the tree bundle. Using the sweep algorithm, we consider at most $n$ starting points and  for each point there are at most $n$ points on the unique path that connects the point to the root. Hence all shortcuts pointing towards the root can be computed in time $\mathcal{O}(n^2)$. In the reverse direction, we exploit the tree graph $G_\textnormal{t}$ as follows. For each point $p$, we compute the shortcuts starting at $p$  by exploring the subtree rooted at $v_p$ in $G_\textnormal{t}$ in a depth-first search manner. For each branching point in the tree, we make a copy of the cone which describes the possible set of shortcut endpoints and operate on that copy in the respective subtree. Apart from that, we proceed exactly as in the classical Chan and Chin algorithm. As each subtree contains at most $n$ nodes and as we need at most one copy of the region at any point, we can indeed compute all shortcuts starting at $p$ in time linear in the total number of points. Hence also the set of shortcuts pointing away from the root can be computed in time $\mathcal{O}(n^2)$.  In conclusion, we can compute the shortcut graph for a given PTB in $\mathcal{O}(n^2)$ when using the Hausdorff distance.
\end{proof}

\begingroup
\def\thetheorem{\ref{THM:correct}}
\thmCorrect*
\endgroup
\begin{proof}
Correctness can be shown by structural induction. The induction hypothesis is that node $v$ gets assigned the optimal simplifcation size $s_{OPT}(v)$ for its sub-tree assuming $v$ is kept in the solution. For leaf nodes (that get assigned a value of 1), correctness is obvious. Now we consider some non-leaf node $v$ and assume that all nodes  $w$ in the respective sub-tree received correct solution size values, that is, $s(w) = s_{OPT}(w)$. Let $C$ be the the set of outgoing shortcut edges emerging from $v$ in an optimal simplification of $Sub(v)$, leading to  an induced simplification size of $s_{OPT}(v) = 1+\sum_{(v,w) \in C} s(w)$. In the DP, each node $w$ with $(v,w) \in C \subseteq E_\textnormal{s}$ gets assigned the helping value $h(w) = s(w)$ when processing $v$.  Based on the propagation of the $h$-value towards the root node, we observe that $h(v)  \leq \sum_{(v,w) \in C} s(w)$ and therefore $s(v) \leq 1+ \sum_{(v,w) \in C} s(w) = s_{OPT}$. We can never get $s(v) < s_{OPT}(v)$ as the value assigned to $v$  always represents a valid simplification of $Sub(v)$. Hence we conclude that $s(v)=s_{OPT}(v)$.

  The time to compute $s(v)$ is in $\mathcal{O}(|Sub(v)|+ |\{(v,w) \in E_s\}|)$ as it requires to consider all nodes in the subtree rooted at $v$ (in post-order) and to check all of its out-neighbors in $G_\textnormal{s}$. Both of these sets have a  size that is upper bounded by the number of nodes in $Sub(v)$, and the post-order can be determined in $\mathcal{O}(|Sub(v)|)$ using  a DFS run. Accordingly, the respective computation time is in  $\mathcal{O}(n)$ for each $v$ and in $\mathcal{O}(n^2)$ for all nodes $v \in V$ combined. The post-order traversal to determine the global order in which the nodes are processed takes only  linear time in $n$  (again using a DFS run). Hence altogether, we have two nested post-order traversals with a total running time of $O(n^2)$. 
\end{proof}
\newpage
\subsection{Additional Experimental Results}
Finally, we provide more detailed experimental results and illustrations of the used benchmark instances.

\subsubsection{Further Tree Bundle Results.}
In Figure \ref{FIG:treeplot}, an example tree bundle instance is depicted along with DP results for a broad range of test instances. It can be  observed that the optimal number of nodes in the simplification converges for growing values of $\delta$ to the number of leaf nodes
in the tree graph, as those have to be kept by definition. But already for small $\delta$, the simplification size comes close to that lower bound.
\begin{figure}
\begin{minipage}{0.6\textwidth}
\includegraphics[width=\textwidth]{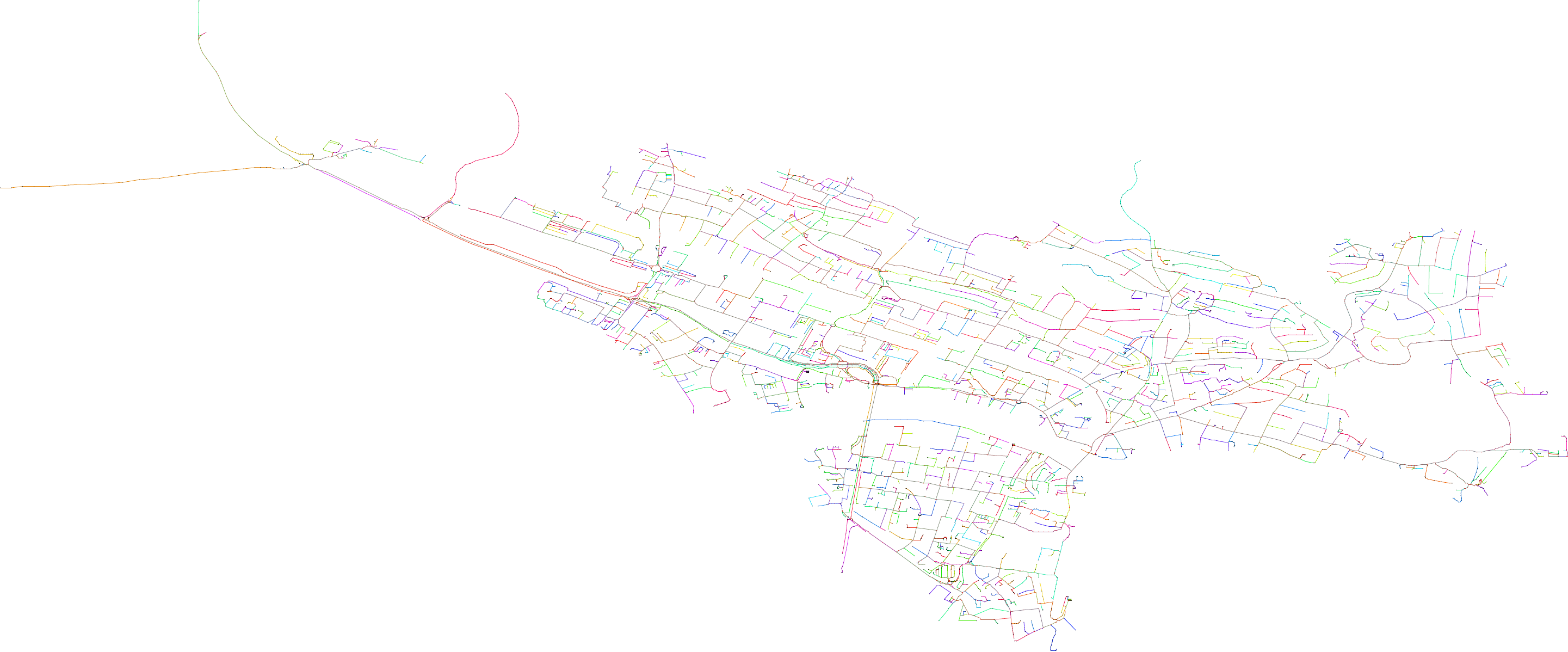}
\end{minipage}\hfill
\begin{minipage}{0.38\textwidth}
\includegraphics[width=\textwidth]{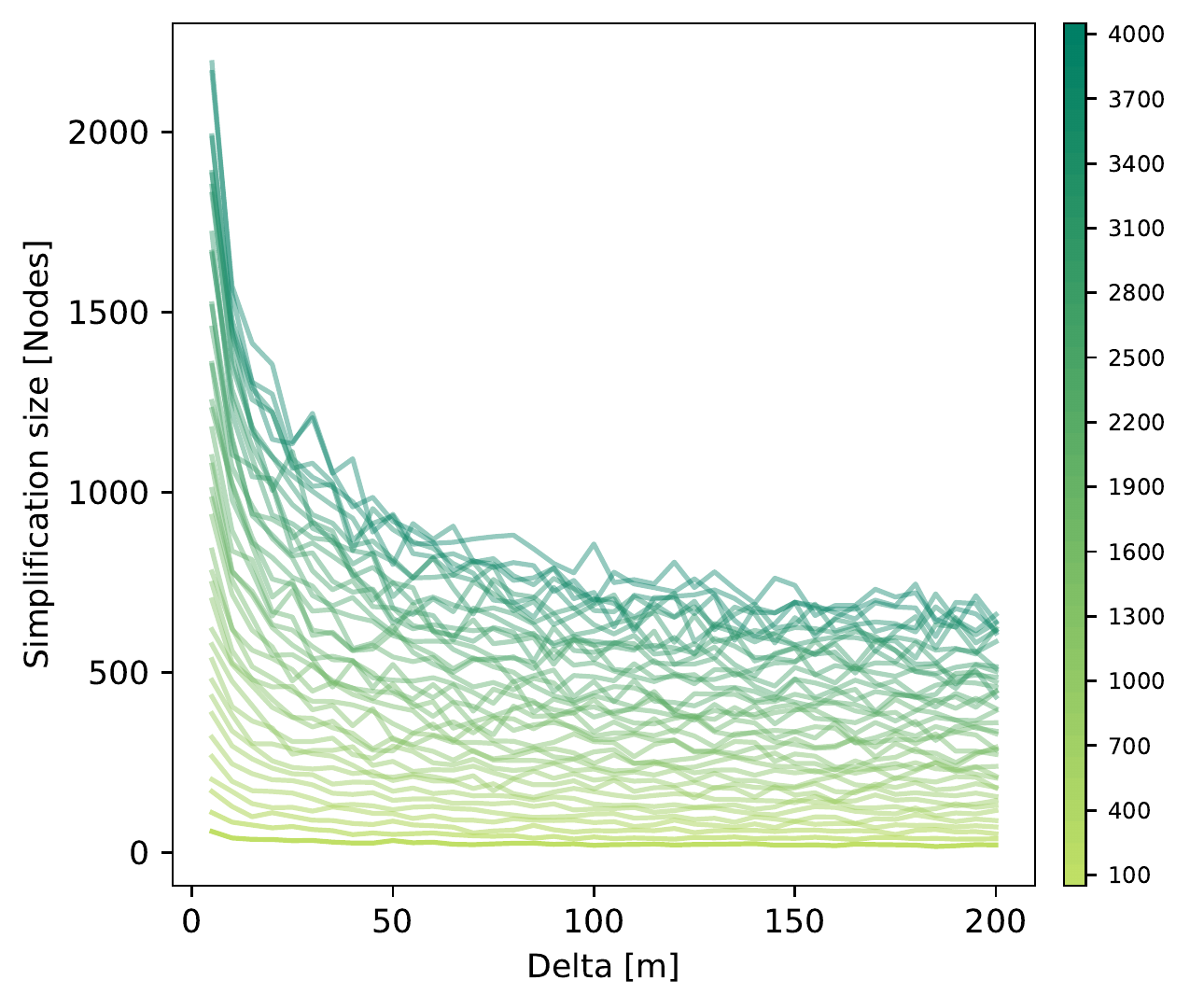}
\end{minipage}
\caption{Left: Example of a road network tree bundle with 1365 polylines.
	Right: Number of retained points in an optimal simplification for tree bundles of different input sizes indicated by line color (legend on the right side) and distance thresholds (on the $x$-axis).
	The largest $\delta$ was chosen so that no further simplification is possible.
	Hence the respective simplification size for $\delta = 200$\,m (the right side of each curve) corresponds to the number of polylines in the input.
	One can nicely see that already for relatively small $\delta$ (around 50\,m), the size of the simplification is close to this minimum.}\label{FIG:treeplot}
\end{figure}
\subsubsection{Detailed TBD Results.} Our greedy tree bundle decomposition heuristic unsurprisingly performed well on inputs with large tree-like structures as Freiburg and Stuttgart but worse on instances as Manhattan or Chicago with large grid-like substructures. For example, the public transit network of Stuttgart with 83 points was decomposed into 12 trees, see Figure \ref{FIG:decomp}, top. The decomposition set consists of 24 nodes of which 22 are endpoints of a line and hence have to be included in the simplification anyway. The individual simplification of the trees for $\delta=0.01$ then added only 12 additional nodes. The public transit network of Freiburg with 785 bends was decomposed in 178 trees, see Figure \ref{FIG:decomp}, bottom. Hereby, 195 nodes in the decomposition stemmed from line endpoints and only 63 additional root nodes were selected. The simplification of the trees then added 103 nodes for $\delta=0.002$ and 339 nodes for $\delta=0.0002$. The network of Chicago with 9,984 points was decomposed into 2649 trees, see Figure \ref{FIG:chicago}. Basically any crossing point in the grid structure was added to $D$, resulting in a   limited simplification capability in the resulting  trees. 

\begin{figure}
\centering
\includegraphics[width=0.4\textwidth]{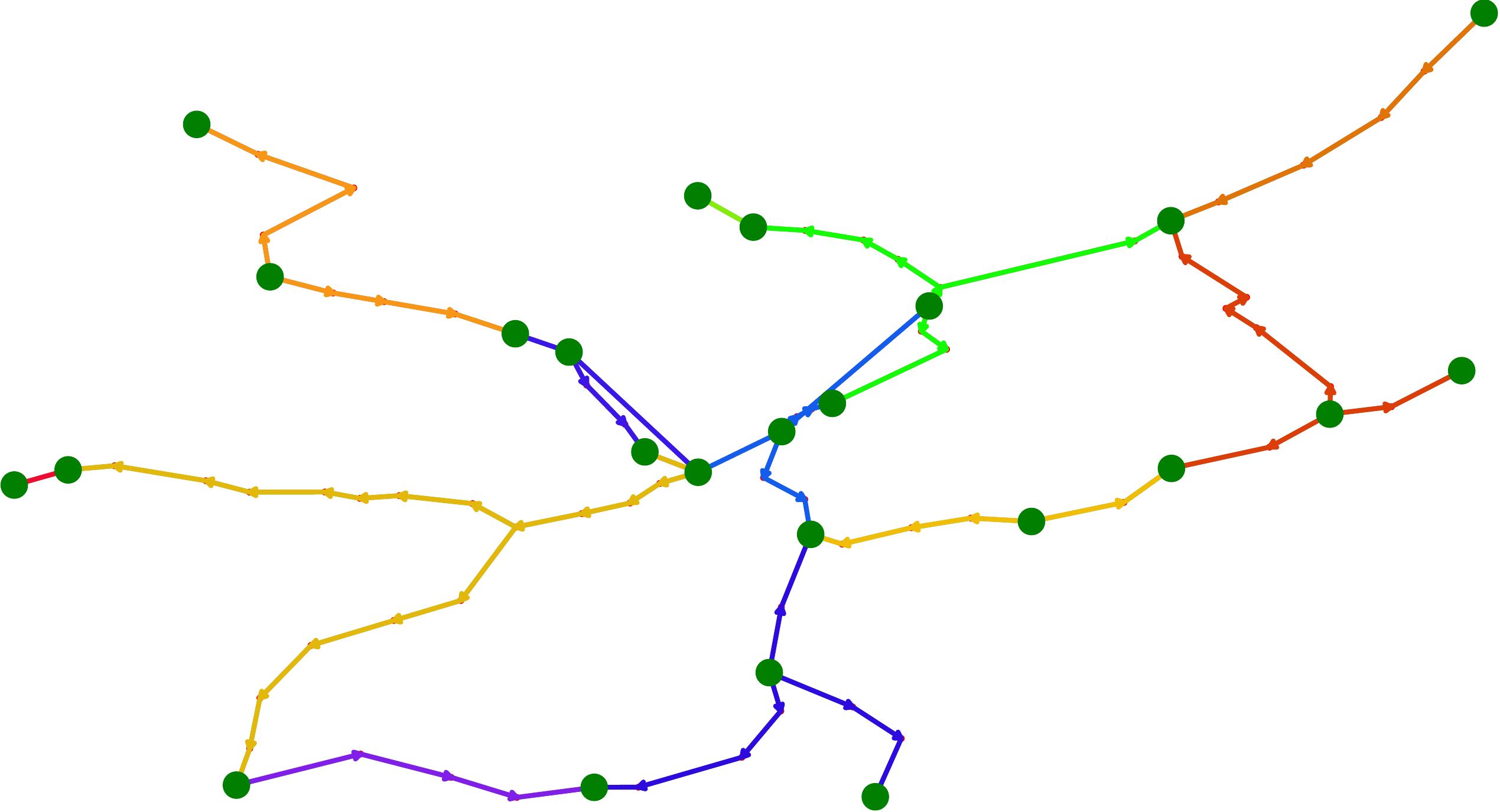}\hfill
\includegraphics[width=0.4\textwidth]{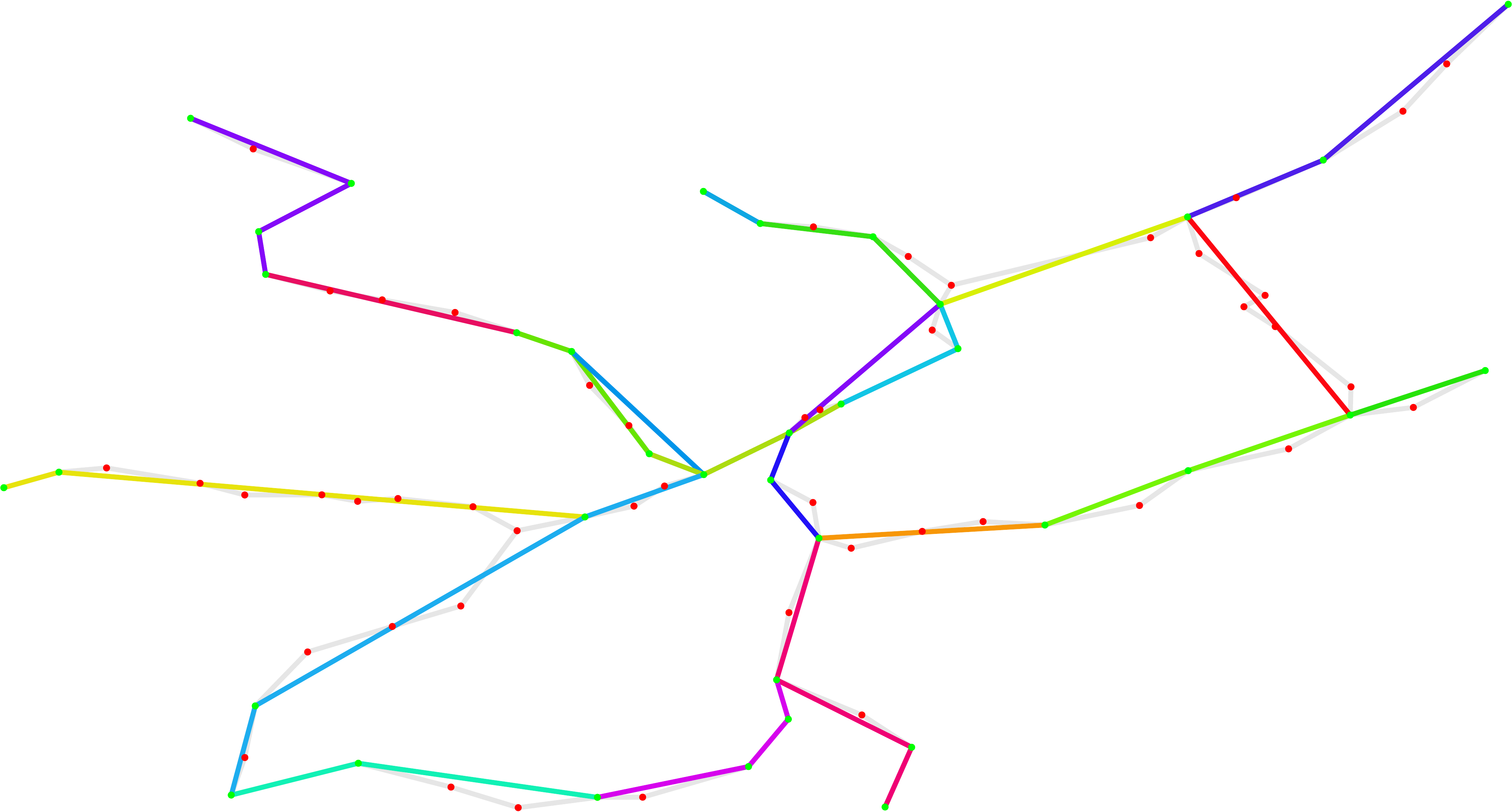}
\bigskip\\
\includegraphics[width=0.8\textwidth]{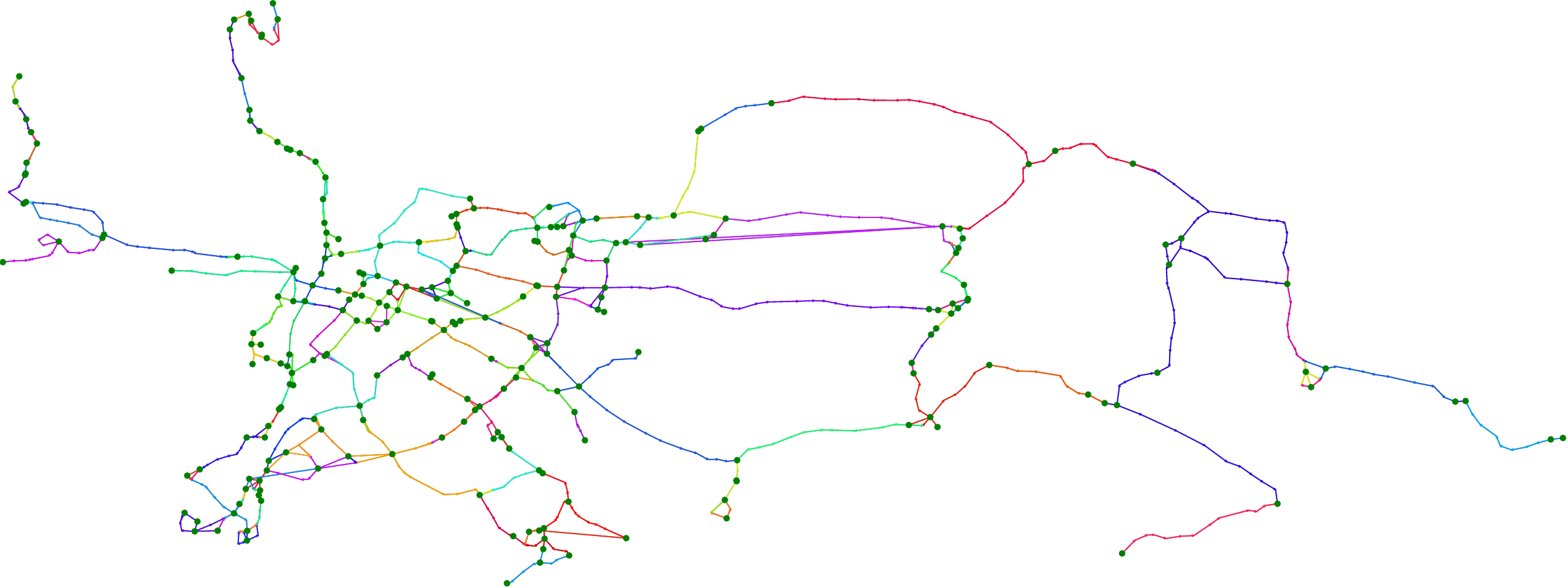}
\caption{Tree decompositions computed with our greedy heuristic. Large green points indicate the nodes in the decomposition set $D$. Each bus and train line received its own random color. Upper row: Decomposition result for the Stuttgart network on the left and resulting simplification on the right (with the original network in the background for comparison). Lower image: Decomposition result for the Freiburg network.}\end{figure}

\begin{figure}
\centering
\includegraphics[width=0.75\textwidth]{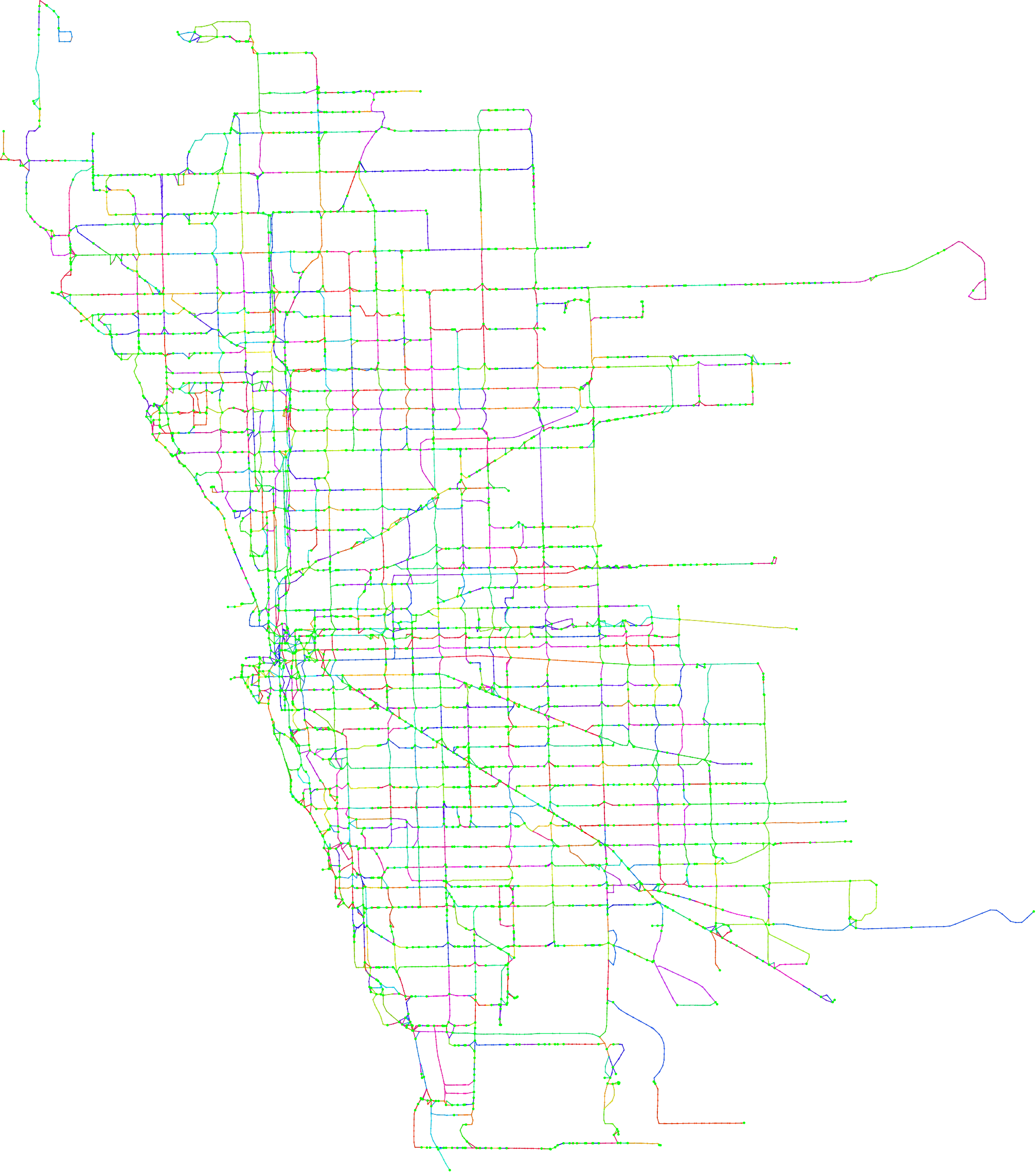}\caption{Chicago public transit network network wth greedy tree decomposition.}\label{FIG:chicago}
\end{figure}\label{FIG:decomp}

\subsubsection{Detailed Comparison of TBD+DP and BCA on General Bundles.}
In Table \ref{TAB:tree-decomp-results}, some results for two of the  transit networks  as well as a large road network bundle for  TBD+DP as well as  BCA are summarized. We observe that TBD+DP is significantly better than BCA on the road network instance in terms of both, quality and running time. But on the Chicago public transit networks, BCA produces the smaller simpification, even if forced to obtain the original $\delta$-constraint (compare the 4601 nodes from TBD+DP to the 1781 from BCA where both solutions obey $\delta \leq 2$).  The reason for this is  the grid-like structure of that network as shown in Figure  \ref{FIG:chicago}. But on other instances, TBD+DP produces high-quality simplification results quickly, as illustrated for a large  road network bundle  in Figure \ref{FIG:bigbundle}.
 \begin{table}
 	\caption{Comparison of BCA and TBD+DP on public transit networks and  a large path bundle extracted from a road network. Timings are in milliseconds.\smallskip\\}
	\label{TAB:tree-decomp-results}
\centering
	\begin{tabular}{|l||l|r|r|r|r|r|r|}
	\hline
	bundle & & $\delta \cdot 10^4$ & $\delta_F \cdot 10^4$ & $\delta_F/\delta$ &  $n+\ell$ & $|S|$ & time \\
	\hline\hline
Stuttgart&	TBD+DP &    10.00   & 9.60  & 0.96     &   83 + 102 &\textbf{36} & 1 \\
	\hline
	&BCA    & 10.00   & 9.60  & 0.96   &   83 + 102   &   {37} & 5  \\\hline\hline
Chicago	&TBD+DP & 2.00  & 1.99 & 0.99  & 9,984 + 872  & {4601}& 306 \\
	\hline
	&TBD+DP& 0.50 & 0.50 & 1.00   & 9,984 + 872      &  {6334}& 185  \\
	\hline
	&BCA   & 2.00  & 3.69 & 1.84    & 9,984 + 872  &  \textbf{864} & 634 \\
	\hline
	&BCA    & 1.00  & 1.92 & 1.92 & 9,984 + 872    &  \textbf{1781}& 789 \\
	\hline
	&BCA    & 0.50 & 0.85& 1.60 & 9,984 + 872     &  \textbf{5023}& 456\\
	\hline\hline
Path bundle	&TBD+DP & 2.00 & 1.96 & 0.98   &10,633 + 17,662 &  \textbf{2681}& 472\\
	\hline
	&BCA   & 2.00  & 2.33 & 1.17   &10,633 + 17,662   & {4915}& 13,700 \\
	\hline
	&BCA    & 1.00  & 1.44 & 1.44  &10,633 + 17,662   & {5017}& 12,600 \\
	\hline
	\end{tabular}
\end{table}

\begin{figure}
\centering
\includegraphics[width=\textwidth]{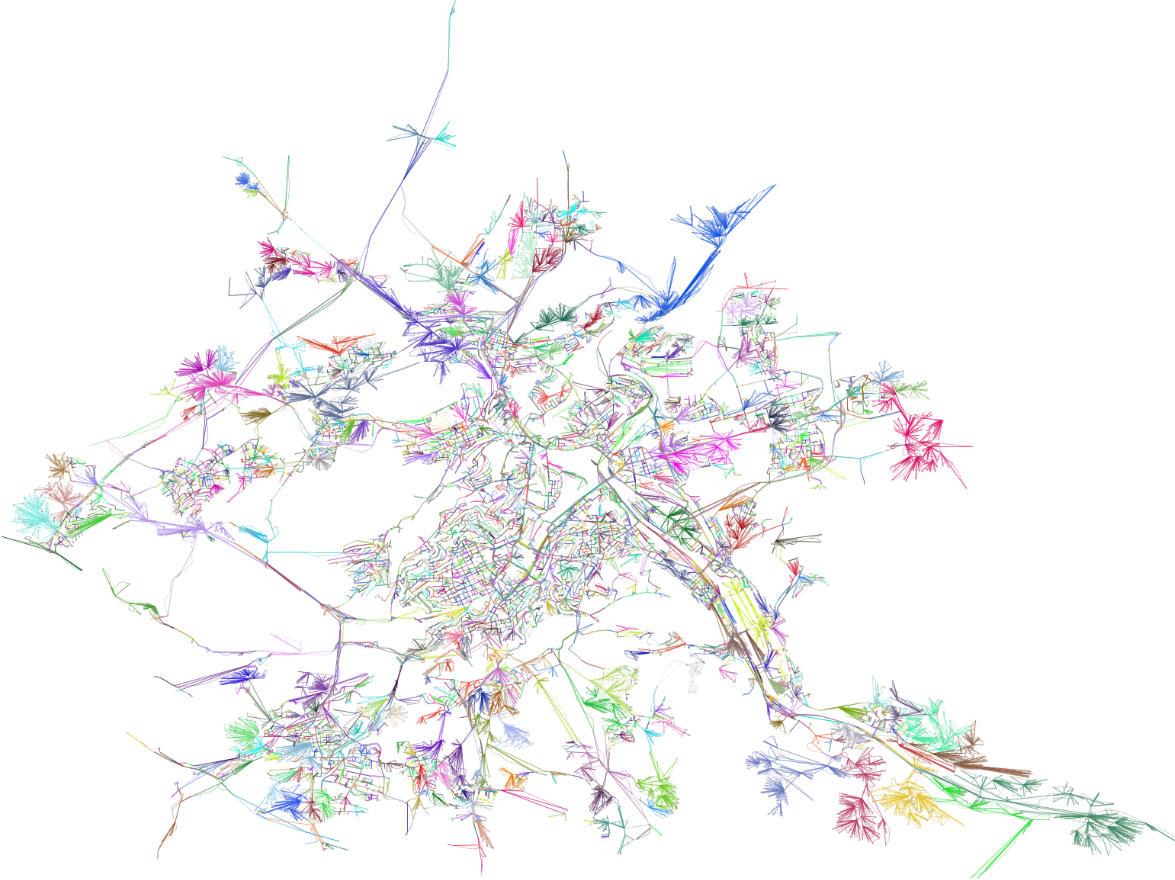}\caption{Large polyline bundle extracted from an underlying road network together with the TBD+DP based simplification.}\label{FIG:bigbundle}
\end{figure}

\end{document}